\documentclass[12pt,a4paper]{article}

\usepackage{amsfonts}
\usepackage{amsmath}
\usepackage{amssymb}
\usepackage{amsthm}

\usepackage{graphicx}

\theoremstyle{plain}
\newtheorem{definition}{Definition}
\newtheorem{theorem}{Theorem}
\newtheorem{lemma}{Lemma}
\newtheorem{corollary}{Corollary}

\theoremstyle{definition}

\numberwithin{equation}{section} \numberwithin{theorem}{section}
\numberwithin{lemma}{section} \numberwithin{definition}{section}

\numberwithin{corollary}{section}

\begin{document}
\title{Existence and stability of stationary solutions to spatially extended autocatalytic and hypercyclic systems under global regulation and with nonlinear growth rates  }

\author{Alexander S. Bratus'$^1$, Vladimir P. Posvyanskii$^1$, Artem S. Novozhilov$^{2,}$\footnote{Corresponding author: novozhil@ncbi.nlm.nih.gov}\\
\textit{\normalsize $^{1}$Moscow State University of Railway Engineering, Moscow, Russia}\\
\textit{\normalsize $^{2}$National Institutes of Health, 8600
Rockville Pike,Bethesda, MD 20894, USA} }

\date{}

\maketitle

\begin{abstract}
Analytical analysis of spatially extended autocatalytic and hypercyclic systems is presented. It is shown that spatially explicit systems in the form of reaction-diffusion equations with global regulation possess the same major qualitative features as the corresponding local models. In particular, using the introduced notion of the stability in the mean integral sense we prove the competitive exclusion principle for the autocatalytic system and the permanence for the hypercycle system. Existence and stability of stationary solutions are studied. For some parameter values it is proved that stable spatially non-uniform solutions appear.

\end{abstract}

\section{Introduction and background}
In 1971 Manfred Eigen published a seminal paper on the evolution
of error-prone self-replicating macromolecules
\cite{eigen1971sma}. His theory was expanded significantly later
on, primarily in works of Eigen, Schuster and co-workers
\cite{eigen1989mcc,eigen1988mqs,eigenshuster}. One of the
principal findings was the existence of the \textit{error
threshold}, i.e., the critical mutation rate such that the
equilibrium population of macromolecules (the
\textit{quasispecies} in the terminology of Eigen et al.) cannot
provide conditions for evolution if the fidelity of copying falls
below this critical level. This critical mutation rate depends on
the length of macromolecules and hence puts limits on the amount
of information that can be carried by a given macromolecule. To
improve fidelity one needs longer sequences (e.g., a more
efficient replicase), to have longer sequences one needs better
fidelity, hence the chicken--egg problem. An easy and obvious
solution to this problem is that the early primordial genomes must
have consisted of independently replicating entities, which,
generally speaking, would compete with each other (see, e.g., \cite{Koonin2005} and references therein).

If we consider a simple mathematical description of independent
competing replicators then the usual differential equations for
the growth take the following form:
\begin{equation}\label{intro:0}
    \frac{\dot{v}_i}{v_i}=a_iv_i^{p}-f_1(t),\quad i=1,\ldots,n,
\end{equation}
where $v_i=v_i(t)$ is the concentration of the $i$-th type of
macromolecules, $a_i$ is the rate of replicating, $p>0$ is the
degree of auto-catalysis, and $f_1(t)$ is the term which is
necessary to keep the total concentration constant, this term
depends only on $t$ and not on the index, $f_1(t)=\sum_{i=1}^n a_iv_i^p
v_i$ in the present case; easy to see that this is equivalent to the condition $\sum_{i=1}^n v_i=1$. In the case if $p\neq 1$ we have system with non-linear growth rates, which model different coupling strength of the various components, for the discussion of such growth rates see, e.g., \cite{mccabe1999etw,metcalf1994owf}. Hereinbelow we consider mainly $p>0$ (or even, $p>1$) but remark that $p=0$ gives
the exponential growth, $p=1$ gives the standard hyperbolic growth
(autocatalysis), and for $p<0$ the parabolic growth occurs
\cite{szathmary1997rrf}. It is straightforward to show that for
$p\geq 0$ only one replicator present at $t\to\infty$, the
competition winds up in the \textit{competitive exclusion} of all
but one types, i.e., the genome composed of independently
replicating entities is not vital.

\begin{figure}[!tb]
\centering
\includegraphics[width=0.4\textwidth]{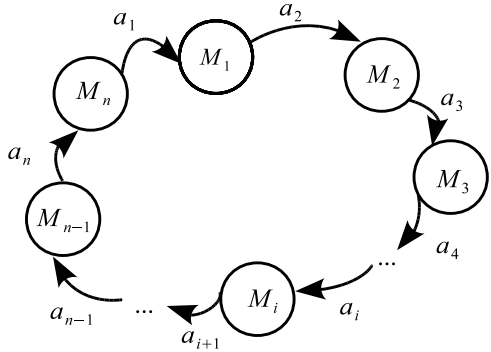}
\caption{The hypercycle \cite{eigenshuster}. Each macromolecule $(M_i)$ helps to replicate another one, $M_{i+1}$, $M_n$ macromolecule promotes the replication of $M_1$ closing the loop; $a_i,\,i=1,\ldots,n$ denote reaction rates.}\label{Fig:intro:1}
\end{figure}
To resolve this situation Eigen and Schuster \cite{eigenshuster}
suggested a concept of the \textit{hypercycle}, a group of
self-replicating macromolecules that catalyze each other in a cyclic
manner: the first type helps the second one, the second type helps the
third, etc, and the last type helps the first one closing the
loop (see Fig. \ref{Fig:intro:1}). An analogue to system \eqref{intro:0} can be written in the
form
\begin{equation}\label{intro:1}
    \frac{\dot{v}_i}{v_i}=a_iv_{i-1}^{p}-f_2(t),\quad i=1,\ldots,n,
\end{equation}
where index $0$ coincides with $n$, $f_2(t)=\sum_{i=1}^n a_iv_{i-1}^p v_i$.
For $p=1$ we obtain the standard hypercycle
model~\cite{hofbauer1998ega}. It is known that \eqref{intro:1} is
\textit{permanent}, i.e., all the concentrations are separated from
zero, and hence different replicators coexist in this model. More
exactly, for short hypercycles, $n=2,3,4$, the internal equilibrium
is globally stable, for longer hypercycles, $n>4$, a globally stable
limit cycle appears~\cite{hofbauer1991sps}.

The problem with the hypercycle model \eqref{intro:1} is its
vulnerability to the invasion of parasites \cite{smith1979hao}.

We remark that models \eqref{intro:0} and \eqref{intro:1} are systems
of ordinary differential equation (ODEs), i.e., they are
mean-field models. As a solution to the parasite invasion problem
it was suggested that heterogeneous population structure can
strengthen persistence of the system. One of the suggested
solution was spatially explicit models
\cite{boerlijst1990ssa,boerlijst1991sws,cronhjort1994hvp}, see
also \cite{boerlijst2000sas,cronhjort2000ibr} for reviews of the
pertinent work. Two major approaches to spatially explicit models
are reaction-diffusion equations and cellular automata models, and
they both were considered in the cited works. Which was lacking,
however, is an analytical treatment of the resulting systems,
because in both cases the researchers have resorted to extensive
numerical simulations. An only notable exception to our knowledge
is \cite{weinberger1991ssa}, where some of the models with
explicit space are analyzed analytically. An interest in
cluster-like solutions of reaction-diffusion systems resulted in
the analysis of spatially explicit hypercycle in infinite space
\cite{wei2000tdr,wei2001hsn}.

Note that models \eqref{intro:0} and \eqref{intro:1} are a special case of the general replicator equation \cite{hofbauer2003egd}, for which several approaches are known to incorporate an explicit spatial structure, albeit there is no universally accepted way of incorporating dispersal effects. The solution to the problem with equal diffusion rates is straightforward, in this case we, following ecological approach, can just add the Laplace operator to the right hand sides of \eqref{intro:0} or \eqref{intro:1}. This was used, e.g., in the classical paper by Fisher \cite{fisher1937waa} to model the effect of the spatial structure on the invasion properties of an advantageous gene; this approach later was generalized by Hadeler \cite{hadeler1981dfs}. However, for the primordial world, it would be a too stringent an assumption to have all the diffusion coefficients equal. To overcome this difficulty, Vickers et al. introduced a special form of the population regulation to allow for different diffusion rates \cite{cressman1997sad, hutson1995sst,vickers1989spa}, now in the subject area of evolutionary game dynamics. In these works a nonlinear term is used that provides \textit{local} regulation of the populations under question, although no particular biological mechanism is known that lets individuals adapt their per capita birth and death rates to local circumstances \cite{ferriere2000ads}. In our view, it is more natural to assume the \textit{global} regulation of the populations, hence following along the lines of thought that brought to the models \eqref{intro:0} and \eqref{intro:1}. Mathematically it means that we assume that the total populations satisfy the following condition
$$
\sum_{i=1}^n\int_{\Omega}v_i(t,x)\,dx=1,
$$
where $x\in\Omega$ is a spatial variable now. This approach was first used in \cite{weinberger1991ssa}. Which is important here is that this approach allows to obtain some analytical insights of the systems \cite{bratus2006ssc}.

In this text our goal is to present an analytical treatment of the models of prebiotic macromolecules with self- and hypercyclic catalysis with an explicit spatial structure and global population regulation in the form of
reaction-diffusion equations.

\section{The mathematical models}
Let $\Omega$ be a bounded domain, $\Omega\subset\mathbb R^{m}$,
$m=1,2,3$, with a piecewise-smooth boundary $\Gamma$. The spatially explicit
analogue to \eqref{intro:0} is given by the following
reaction-diffusion system
\begin{equation}\label{model:1}
    \partial_t v_i=v_i(a_iv_i^{p}-f_1(t)
    )+d_i\Delta v_i,\quad i=1,\ldots,n,\quad t>0.
\end{equation}
Here $v_i=v_i(x,t),\,x\in\Omega,t>0$, $\partial_t\equiv \frac{\partial}{\partial t}$, $\Delta$ is the Laplace
operator, in the Cartesian coordinates
$\Delta=\sum_{k=1}^m\frac{\partial^2}{\partial x_i^2}$. The initial
conditions are $v_i(x,0)=\varphi_i(x)$, $p>0$ (although we note, that in each particular case we shall specify admissible values of $p$), and the form of
$f_1(t)$ will be determined later.

A slight modification of \eqref{model:1} gives the hypercyclic
system
\begin{equation}\label{model:2}
    \partial_t v_i=v_i(a_iv_{i-1}^{p}-f_2(t)
    )+d_i\Delta v_i,\quad i=1,\ldots,n,\quad t>0,
\end{equation}
where $v_0\equiv v_n$.

In both problems \eqref{model:1} and \eqref{model:2} the functions
$v_i(x,t)$ are assumed to be nonnegative, since they represent
\textit{relative} concentrations of different macromolecules.

It is natural to assume that we consider closed systems (see also
\cite{weinberger1991ssa}), i.e., we have the boundary conditions
\begin{equation}\label{bnd:1}
    \left.\frac{\partial v_i}{\partial
    \textbf{n}}\right|_{x\in\Gamma}=0,\quad i=1,\ldots,n,
\end{equation}
where $\textbf{n}$ is the normal vector to the boundary $\Gamma$.

It is assumed that the global regulation of the total
concentration of macromolecules occurs in the system such that
\begin{equation}\label{2:1}
\sum_{i=1}^n \int_{\Omega}v_i(x,t)\,dx=1
\end{equation}
for any time moment $t$. This condition is an analogous condition
for the total concentration of replicators in the finite-dimensional
case \cite{hofbauer1998ega}. From the boundary condition \eqref{bnd:1}
and the integral invariant \eqref{2:1} the expressions for the
functions $f_1(t)$ and $f_2(t)$ follow:
\begin{equation}\label{inv:1}
    f_1(t)=\sum_{i=1}^n
    \int_{\Omega}a_iv_i^{p+1}(x,t)\,dx
\end{equation}
and
\begin{equation}\label{inv:2}
    f_2(t)=\sum_{i=1}^n
    \int_{\Omega}a_iv_{i-1}^p(x,t)v_i(x,t)\,dx.
\end{equation}
Finally we have a mixed problem for a system of semilinear parabolic
equations with the integral invariant \eqref{2:1} and functionals
\eqref{inv:1} and \eqref{inv:2}.

Suppose that for any fixed $x\in\Omega$ each function $v_i(x,t)$ is
differentiable with respect to variable $t$, and belongs to the
space $H_{p+1}^1(\Omega)$ as the function of $x$ for any fixed
$t>0$. Here $H_{p+1}^1$ is the space of functions with the norm
$$
\|u(x)\|_{H_{p+1}^1}=\left[\int_{\Omega}|u(x)|^{p+1}\,dx\right]^{\frac{1}{p+1}}+\left[\int_{\Omega}\sum_{k=1}^m\left|\frac{\partial
u}{\partial x_k}\right|^{2}\,dx\right]^{\frac{1}{{2}}}.
$$
Note that if $p\geq 1$ then $H_{p+1}^1(\Omega)\subseteq
H_{2}^1(\Omega)$, where $H_{2}^1(\Omega)$ is the Sobolev space of
square-integrable functions for which their first partial
derivatives are also square-integrable \cite{mikhlin1964vmm}.

Without loss of generality we shall assume further that volume of
the domain $\Omega$ is equal 1: $|\Omega|=1$.

Our main goal is to analyze existence and stability of the steady
state solutions to \eqref{model:1} and \eqref{model:2}. The steady
state solutions are given by the solutions to the following
elliptic problems:
\begin{equation}\label{model:1:1}
    d_i\Delta u_i+u_i(a_iu_i^{p}-\bar{f}_1)=0,\quad
    i=1,\ldots,n,
\end{equation}
and
\begin{equation}\label{model:1:2}
    d_i\Delta u_i+u_i(a_iu_{i-1}^{p}-\bar{f}_2)=0,\quad
    i=1,\ldots,n,\quad u_0\equiv u_n,
\end{equation}
with the boundary conditions $\partial_{\textbf{n}}u_i=0$ on
$\Gamma$; $u_i(x)\in H_{p+1}^1(\Omega)$. The integral invariant
\eqref{2:1} now reads
\begin{equation}\label{2:2}
    \sum_{i=1}^n\int_{\Omega}u_i(x)\,dx=1,
\end{equation}
the values of $\bar{f}_1$ and $\bar{f}_2$ are constant:
\begin{equation}\label{inv:1:1}
    \bar{f}_1=\sum_{i=1}^n
    \int_{\Omega}a_iu_i^{p+1}(x)\,dx
\end{equation}
and
\begin{equation}\label{inv:2:1}
    \bar{f}_2=\sum_{i=1}^n
    \int_{\Omega}a_iu_{i-1}^p(x)u_i(x)\,dx.
\end{equation}

If it is assumed that $d_1=d_2=\ldots =d_n=0$ then the equilibrium
points of \eqref{intro:0} and \eqref{intro:1} coincide with the
steady state solutions to \eqref{model:1} and \eqref{model:2}.
These solutions are spatially homogeneous. The converse is also
true: the spatially homogeneous equilibria of systems
\eqref{model:1} and \eqref{model:2} are fixed points of the
dynamical systems \eqref{intro:0} and \eqref{intro:1}
respectively.

The coordinates of these spatially homogeneous solutions are
straightforward to write down. Let $\beta_i=(a_i)^{-\frac{1}{p}}$ and consider the sum $\beta=\sum
\beta_i$, where the index of summation is determined later. All
spatially homogeneous solutions to \eqref{model:1:1} are given by
\begin{equation*}
    \begin{split}
    P_1 & =\frac{1}{\beta}(\beta_1,\,\beta_2,\ldots,\beta_n),\\
      Q_j  &
      =\frac{1}{\beta}(\beta_1,\ldots,\beta_{j-1},0,\beta_{j+1},\ldots,\beta_n),\\
      Q_{jk}  &
      =\frac{1}{\beta}(\beta_1,\ldots,\beta_{j-1},0,\beta_{j+1},\ldots,\ldots,\beta_{k-1},0,\beta_{k+1},\ldots,\beta_n),\\
      &\ldots
\end{split}
\end{equation*}
ending with the vertices $R_i=(0,\ldots,0,1,0,\ldots,0)$ (unity at
the $i$-th place) of the simplex $\sum_{i=1}^n
u_i=1,\,u_i\geq 0$; for each steady state $\beta$ in obtained
by summing through all non-zero elements in the vector.

The spatially homogeneous stationary solution to \eqref{model:1:2}
is given by
$$
P_2=\frac{1}{\beta}(\beta_2,\,\beta_3,\ldots,\beta_n,\,\beta_1).
$$
\section{Stability of spatially homogeneous equilibria}
Let $u^0=(u_1^0,\ldots,u_n^0)$ be a spatially homogeneous solution
to system \eqref{model:1}. In the usual way we assume that the
Cauchy data are perturbed
$$
\varphi_i(x)=u_i^0+w_i(x),\quad i=1,\ldots,n.
$$
Here $w_i(x)\in H_{p+1}^1(\Omega)$. Inasmuch as we have
$$
\sum_{i=1}^nu_i^0=1,
$$
then from \eqref{2:1} it follows that
\begin{equation}\label{3:1}
    \sum_{i=1}^n \int_{\Omega} w_i(x)\,dx=0.
\end{equation}

Consider the following eigenvalue problem
\begin{equation}\label{eigenv:pr}
\Delta \psi(x)+\lambda\psi(x)=0,\quad
x\in\Omega,\quad\partial_{\textbf{n}}\psi(x)|_{x\in\Gamma}=0.
\end{equation}
The system of eigenfunctions of this problem $\psi_0(x)=1$,
$\{\psi_i(x)\}_{i=1}^{\infty}$ forms a complete system in the Sobolev
space $H_1^1(\Omega)$ \cite{mikhlin1964vmm} such that
$$
\langle\psi_i(x),\psi_j(x)\rangle=\int_{\Omega}\psi_i(x)\psi_j(x)\,dx=\delta_{ij},
$$
where $\delta_{ij}$ is the Kronecker symbol. The corresponding
eigenvalues satisfy the condition
$$
0=\lambda_0<\lambda_1\leq\lambda_2\leq\ldots\leq\lambda_i\leq\ldots,\quad
\lim_{i\to\infty}\lambda_i=\infty.
$$
Hence for $p\geq1$ we assume that $w_i(x)$ can be represented
as
\begin{equation}\label{3:2}
    w_i(x)=\sum_{j=0}^{\infty}c_j^i\phi_j(x),\quad i=1,\ldots,n,
\end{equation}
where $c_j^i$ are constant.

Denote $H_{\delta}$ the set of functions $w(x)\in H_1^1(\Omega)$
such that $\|w\|_{H_1^1}\leq \delta$, where $\delta>0$.
\begin{theorem}\label{th3:1}
For $p\geq1$ all spatially homogeneous stationary solutions to
\eqref{model:1} are unstable with respect to any perturbation from
the set $H_\delta$ if
\begin{equation}\label{3:3}
    0<\frac{d_i}{a_i}<\frac{p}{\lambda_1},\quad i=1,\ldots,n.
\end{equation}

The solutions $R_i=(0,\ldots,0,1,0,\ldots,0)$, $i=1,\ldots,n$ are
stable when
\begin{equation}\label{3:4}
\frac{d_i}{a_i}>\frac{p}{\lambda_1}\,.
\end{equation}
Here $\lambda_1$ is the first non-zero eigenvalue of the problem
\eqref{eigenv:pr}.
\end{theorem}
\begin{proof}
Let $W(x,t)=(W_1(x,t),\ldots,W_n(x,t))$ be a vector-function
belonging to $H_\delta$ for any fixed $t$. Using \eqref{3:1} and
\eqref{eigenv:pr} we can seek the solution to \eqref{model:1} in the
following form:
\begin{equation}\label{3:6}
    v_i(x,t)=u_i^0+W_i(x,t),\quad
    W_i(x,t)=\sum_{j=0}^nc_j^i(t)\psi_j(x).
\end{equation}
Substituting \eqref{3:6} into \eqref{model:1} and retaining in the
usual way only linear terms with respect to $W_i$ we obtain the
following equations:
\begin{equation}\label{3:7}
    \partial_tW_i=(p+1)a_i(u_i^0)^pW_i-\bar{f}_1W_i+d_i\Delta
    W_i,\quad W_i(x,0)=\varphi_i(x)\in H_\delta,\quad \partial_{\textbf{n}}
    W_i=0\mbox{ on }\Gamma.
\end{equation}

Consider first the case $u^0=P_1$. Direct calculations show that
$\bar{f}_1=\beta^{-p}$.

Multiplying equations \eqref{3:7} one after another by the functions
$\psi_j$ and integrating with respect to $x\in\Omega$ we obtain the following system of
ordinary differential equations:
\begin{equation}\label{3:8}
    \frac{dc_j^i(t)}{dt}=c_j^i(t)(a_ip\beta^{-p}-d_i\lambda_j),\quad
    i=1,\ldots,n,\quad j=0,1,2,\ldots
\end{equation}
For $j=0$ one has
$$
c_0^i(t)=c_0^i(0)\exp(a_ip\beta^{-p}t),
$$
therefore, $c_0^i(t)\to\infty$ as $t\to\infty$, which implies that
$P_1$ is unstable.

Using the same approach it is straightforward to show that
$Q_j,\,Q_{jk},\ldots$ are also unstable.

Now we deal with $R_i$. First note that from \eqref{2:2} it follows
that
\begin{equation}\label{3:9}
    \sum_{i=1}^nc_0^i(t)=0.
\end{equation}

For $R_i$ we have
$$
\frac{dc_j^i(t)}{dt}=-c_j^i(t)(a_i+d_i\lambda_j),\quad\mbox{for }j\neq
i
$$
and
$$
\frac{dc_j^i(t)}{dt}=c_j^i(t)(a_i-d_i\lambda_j),\quad\mbox{for }j=i.
$$
Therefore, for $j\neq i$, $c_i^j(t)\to0$ when $t\to\infty$. Taking
into account \eqref{3:9} we obtain that $c_0^i(t)\to0$. If $i=j$ and
\eqref{3:3} holds then $c_j^i(t)\to\infty$, if \eqref{3:4} holds then
$c_j^i(t)\to 0$, which proves the theorem.
\end{proof}

\begin{theorem}
If $p\geq 1$ then spatially homogeneous stationary solution
$P_2$ to system \eqref{model:2} is unstable with respect to any
perturbations from the set $H_{\delta}$ when
\begin{equation}\label{3:10}
    \prod_{i=1}^n\frac{d_i}{a_i}<\left(\frac{p}{\beta^p\lambda_1}\right)^n.
\end{equation}
\end{theorem}
\begin{proof}
As before we will look for a solution to \eqref{model:2} in the form
\eqref{3:6}. After substituting \eqref{3:6} into \eqref{model:2},
multiplying by $\psi_j$ and integrating, we obtain the following
system of ordinary differential equations for $c_j^i(t)$:
\begin{equation}\label{3:12}
    \frac{dc_j^i(t)}{dt}=\frac{p}{\beta^p}\frac{a_i}{a_{i+1}}c_{j}^{i-1}(t)-d_i\lambda_jc_j^i(t),\quad
    i=1,\ldots n,\,n+1\equiv 1,\,0\equiv n,\quad j=0,1,2,\ldots
\end{equation}
Applying the Routh--Hurwitz criterion we obtain that the solutions
to \eqref{3:12} go to $\infty$ if \eqref{3:10} holds, which implies
instability of $P_2$.
\end{proof}
\paragraph{Remark 3.1.} Inverse inequality to \eqref{3:10} provides
stability of $P_2$ only in the cases $n=2,3,4$. Actually, for $j=0$
we have that \eqref{3:12} takes the form
\begin{equation*}
    \frac{dc_0^i(t)}{dt}=\frac{p}{\beta^p}\frac{a_i}{a_{i+1}}c_{0}^{i-1}(t),\quad
    i=1,\ldots, n.
\end{equation*}
All eigenvalues can be easily evaluated because the corresponding
matrix is circular:
$$
\mu_j=\frac{p}{\beta^p}\rho_j,\quad j=0,\ldots,n-1,
$$
where $\rho_j$ is the $j$-th root of the equation $\rho^n=1$. The
eigenvector $(1,1,\ldots,1)$ does not satisfy \eqref{3:9}, therefore we exclude it from the consideration. When
$n=2,3$ all eigenvalues have negative real part, in the case $n=4$
$P_2$ also will be stable \cite{hofbauer1998ega}. For $n\geq 5$
there is at least one eigenvalue with positive real part, which
proves the claim that $P_2$ is unstable when $n\geq 5$.

\section{Existence of spatially nonuniform stationary solutions to  systems \eqref{model:1} and \eqref{model:2} in one-dimensional case}
Here we will prove that when the space is one dimensional,
$\Omega=[0,1],\,\Delta=\partial_x$, the models \eqref{model:1} and
\eqref{model:2} possess non-uniform stationary solutions under some
additional conditions. The boundary conditions now take the form
$\partial_x v_i(0,t)=\partial_x v_i(1,t)=0$.

\begin{theorem}\label{th4:1}
For $0<p\leq 2$ a spatially non-uniform stationary solution to
\eqref{model:1} exists if the following inequality holds
\begin{equation}\label{4:1}
    \sum_{i=1}^n\left(\frac{d_i}{a_i}\right)^{\frac{1}{p}}<\left(\frac{p}{\pi^2}\right)^{\frac{1}{p}}.
\end{equation}
\end{theorem}
\begin{proof}
We start the proof noting that the dependence of the concentrations
$u_i(x)$ in \eqref{model:1:1} on other concentrations and their
total regulations occur only through the integral invariant
\eqref{inv:1:1}, which does not depend on $x$. Therefore we can assume
without loss of generality that each $u_i$ depends on its own
variable $x_i\in[0,1]$. Hence we rewrite \eqref{2:2} and
\eqref{inv:1:1} in the form
\begin{equation}\label{4:2}
\begin{split}
    &\sum_{i=1}^n\int_{\Omega}u_i(x_i)\,dx_i=1,\\
    \bar{f}_1&=\sum_{i=1}^n    \int_{\Omega}a_iu_i^{p+1}(x_i)\,dx_i.
\end{split}
\end{equation}
Each equation of system \eqref{model:1:1} can be put in the
following form:
\begin{equation}\label{4:3}
    \begin{split}
    \frac{du_i}{dx_i} &=V_i,\\
     \frac{dV_i}{dx_i}   &=\frac{1}{d_i}(\bar{f}_1-a_iu_i^p)u_i.
\end{split}
\end{equation}
System \eqref{4:3} is a Hamiltonian system for any $i=1,\ldots,n$,
in which $x_i$ is considered as a ``time'' variable, with the
Hamiltonian
$$
H_i=\frac{V_i^2}{2}+\frac{1}{d_i}\left(\frac{d_i}{p+2}u_i^{p+2}-\frac{\bar{f}_1}{2}u_i^2\right).
$$
The phase orbits of \eqref{4:3} can be found from the standard formula
$$
V_i=\pm\sqrt{2(H_i^0-U_i)},
$$
where
\begin{align*}
H_i^0&=\frac{V_i^2(x_i^0)}{2}+\frac{1}{d_i}\left(\frac{d_i}{p+2}u_i^{p+2}(x_i^0)-\frac{\bar{f}_1}{2}u_i^2(x_i^0)\right),\\
U_i&=\frac{1}{d_i}\left(\frac{d_i}{p+2}u_i^{p+2}-\frac{\bar{f}_1}{2}u_i^2\right).
\end{align*}

\begin{figure}
\centering
\includegraphics[width=0.7\textwidth]{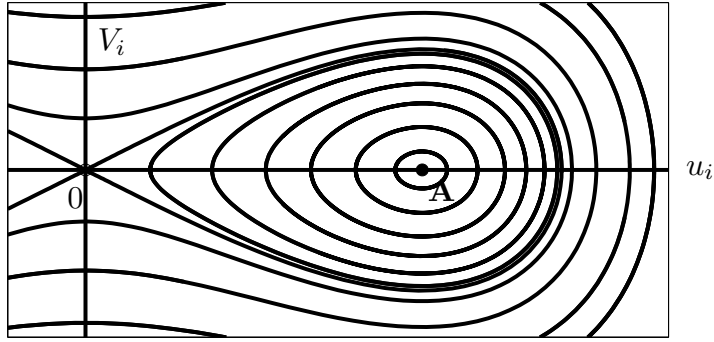}
\caption{The phase portrait of system \eqref{4:3}. The closed curves
surrounding $\bf{A}$ correspond to the candidates for spatially
non-homogeneous stationary solutions to \eqref{model:1}}\label{f1}
\end{figure}

From the form of the phase orbits (see Fig. \ref{f1}) it immediately
follows that there exist orbits that satisfy the condition
$$
V_i(x_i^1)=V_i(x_i^2)=0,\quad x_i^1\neq x_i^2.
$$
These orbits represent closed curves surrounding the center point
$\textbf{A}=\left(\left({\bar{f}_1}{a_i^{-1}}\right)^\frac{1}{p},\,0\right)$
in Fig. \ref{f1}. Different diffusion coefficients correspond to
the motion along the phase orbits with different velocities.

To prove the theorem we need to show that there exist two values
$x_i^1$ and $x_i^2$ such that $|x_i^1-x_i^2|$=1, for
$i=2,\ldots, n$, and corresponding solutions to \eqref{model:1:1}
satisfy the first condition in \eqref{4:2}.

The solutions to system \eqref{model:1:1} can be found in the explicit
parametric form \cite{polyanin1995hes}:
\begin{equation}\label{4:4}
    \begin{split}
    u_i(x_i) & =\left[\frac{p+2}{2a_i}\bar{f}_1\right]^\frac{1}{p}\tau,\quad \tau\geq 0,\\
        x_i  & =\sqrt{\frac{d_i}{\bar{f}_1}}\int_{\tau_0}^\tau
        \frac{dt}{\sqrt{P_i(t)}}+c_i^2,\quad
        P_i(t)=c_i^1+t^2-t^{p+2},\\
            u_i(x_i) & =-\left[\frac{p+2}{2a_i}\bar{f}_1\right]^\frac{1}{p}\tau,\quad \tau\leq 0,\\
        x_i  & =\sqrt{\frac{d_i}{\bar{f}_1}}\int_{\tau_0}^\tau
        \frac{dt}{\sqrt{Q_i(t)}}+c_i^2,\quad
        Q_i(t)=c_i^1+t^2+t^{p+2}.
\end{split}
\end{equation}

To proceed we need the following lemma (the proof is given
in the Appendix).
\begin{lemma}\label{l4:1}
The equation $P_i(t)=c_i^1+t^2-t^{p+2}=0$ has two real positive
roots $0<\tau_i^1<\tau_i^2$ for all values
\begin{equation}\label{4:5}
    c_i^1\in\left(-\left[\frac{2}{2+p}\right]^{\frac{2}{p}+1},\,0\right).
\end{equation}
Moreover,
\begin{equation}\label{4:6}
    \tau_i^1\in\left(0,\,\left[\frac{2}{2+p}\right]^{\frac{2}{p}}\right),\quad
    \tau_i^2\in\left(\left[\frac{2}{2+p}\right]^{\frac{2}{p}},\,1\right).
\end{equation}
\end{lemma}
An analogous lemma holds for $Q_i(t)$.

Now we return to the parametric representation \eqref{4:4}. Consider the
first derivative of the functions $u_i$:
$$
\frac{du_i}{dx_i}=\frac{du_i}{d\tau}\frac{d\tau}{dx_i}=\left[\frac{p+2}{2a_i}\bar{f}_1\right]^\frac{1}{p}\sqrt{\frac{d_i}{\bar{f}_1}}\sqrt{P_i(t)}.
$$
This expression vanishes at the points $\tau_i^1$ and $\tau_i^2$.
Using \eqref{4:4} for $x_i=0$ and $x_i=1$, we obtain
\begin{equation}\label{4:7}
\begin{split}
    x_i=0 &\longrightarrow\quad\sqrt{\frac{d_i}{\bar{f}_1}}\int_{\tau_0}^{\tau_i^1}\frac{dt}{\sqrt{P_i(t)}}+c_i^2=0\\
    x_i=1 &\longrightarrow\quad\sqrt{\frac{d_i}{\bar{f}_1}}\int_{\tau_0}^{\tau_i^2}\frac{dt}{\sqrt{P_i(t)}}+c_i^2=1.
\end{split}
\end{equation}
Letting $\tau_0=\tau_i^1$ we obtain that $c_i^2=0$. We will use the
following notation:
\begin{equation}\label{4:8}
I_i^1=\int_{\tau_i^1}^{\tau_i^2}\frac{dt}{\sqrt{P_i(t)}}=\sqrt{\frac{\bar{f}_1}{d_i}}.
\end{equation}
Remark that these integrals will exist because the roots of $P_i(t)$
are simple when $c_i^1$ satisfy \eqref{4:5}. The formula \eqref{4:8}
establishes the connection between the values of the constant
$c_i^1$ and values of the diffusion coefficient $d_i$, which
determines the velocity of motion of phase points. The latter
implies that \eqref{4:8} guaranteers that the motion from the
initial point $(u_i^1,0)$ to the final point $(u_i^2,0)$ occurs
during the unit time.

Now we are going to prove that the solution \eqref{4:4} satisfies
the first condition in \eqref{4:2}:
\begin{equation}\label{4:9}
    \sum_{i=1}^n\int_{0}^1u_i(x_i)\,dx_i=\left[\frac{p+2}{2}\right]^\frac{1}{p}\sum_{i=1}^n\left[\frac{\bar{f}_1}{a_i}\right]^{\frac{1}{p}}\sqrt{\frac{d_i}{\bar{f}_1}}\int_{\tau_i^1}^{\tau_i^2}
        \frac{tdt}{\sqrt{P_i(t)}}=1.
\end{equation}
From \eqref{4:9} it follows that
\begin{equation*}
    \left[\frac{p+2}{2}\right]^\frac{1}{p}\sum_{i=1}^n\left[\frac{\bar{f}_1}{a_i}\right]^{\frac{1}{p}}\sqrt{\frac{d_i}{\bar{f}_1}}\int_{\tau_i^1}^{\tau_i^2}
        \frac{(p+2)t^{p+1}/2-tdt}{\sqrt{P_i(t)}}=0.
\end{equation*}
Indeed, we have

\begin{equation*}
\begin{split}
&\int_{\tau_i^1}^{\tau_i^2}\frac{(p+2)t^{p+1}/2-tdt}{\sqrt{P_i(t)}}=\frac{1}{2}\int_{\tau_i^1}^{\tau_i^2}
       \frac{d( t^{p+2}-t^2)}{\sqrt{P_i(t)}}=\\
       &= -\frac{1}{2}\int_{\tau_i^1}^{\tau_i^2}\frac{d(
       P_i(t))}{\sqrt{P_i(t)}}=-\left(\sqrt{P_i(\tau_i^2)}-\sqrt{P_i(\tau_i^1)}\right)=0.
\end{split}
\end{equation*}

Using \eqref{4:8} to find $\bar{f}_1$ and substituting this
expression into \eqref{4:9} we obtain
\begin{equation}\label{4:11}
    \left[\frac{p+2}{2}\right]^{\frac{1}{p}}\sum_{i=1}^n\left[\frac{d_i}{a_i}\right]^\frac{1}{p}\left(I_i^1\right)^{\frac{2}{p}-1}\left(I_i^2\right)=1,
\end{equation}
where
$$
I_i^1=\int_{\tau_i^1}^{\tau_i^2}\frac{dt}{\sqrt{P_i(t)}},\quad
I_i^2=\int_{\tau_i^1}^{\tau_i^2}\frac{tdt}{\sqrt{P_i(t)}}.
$$
To conclude the proof we need the following lemma, the proof of which is
given in the Appendix
\begin{lemma}\label{l4:2}
If $0<p\leq 2$ then the following inequality holds:
\begin{equation}\label{4:12}
    I(c_i^1)=(I_i^1)^{\frac{2}{p}-1}\left[\frac{p+2}{2}\right]^{\frac{1}{p}}I_i^2>\left[\frac{\pi^2}{p}\right]^\frac{1}{p}.
\end{equation}
\end{lemma}

Applying the result of Lemma \ref{l4:2} to \eqref{4:11} we obtain
that if \eqref{4:1} holds then there exists a spatially non-uniform
stationary solution to \eqref{model:1}.
\end{proof}

\paragraph{Remark 4.1.}From the symmetry of the system, the time
needed to get from the point $(u_i^1,0)$ to the point $(u_i^2,0)$ is
the same as the time needed to get from $(u_i^2,0)$ to $(u_i^1,0)$,
and the speed of movement is inversely proportional to $\sqrt{d_i}$.
Therefore, reducing these values twice we guaranteer that spatially
non-uniform stationary solution exists, which corresponds to the
full cycle in the phase plane; reducing 4 times we obtain the
solution which corresponds to the movement of the phase point along
the cycle two times, and so on. Hence system \eqref{model:1} has
non-uniform stationary solutions that correspond to movement along
the cycles in Fig. \ref{f1} arbitrary number of times (see Fig. \ref{f2}).

\begin{figure}
\centering
\includegraphics[width=0.8\textwidth]{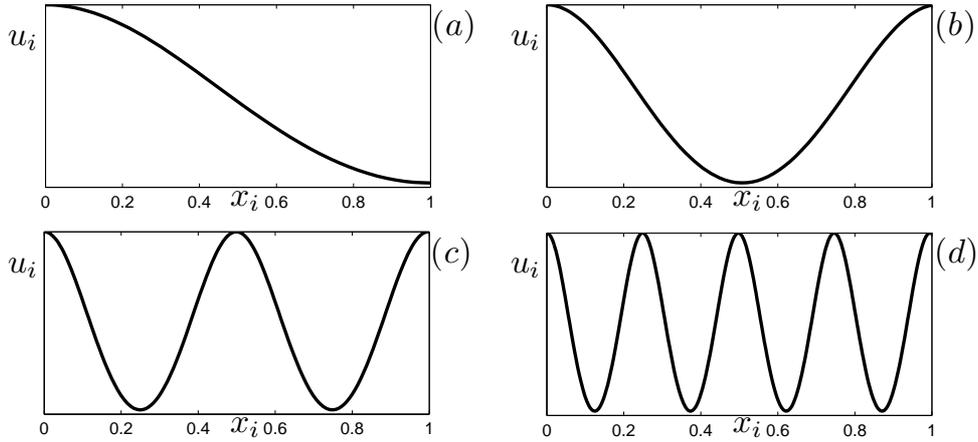}
\caption{Spatially non-uniform stationary solutions to
\eqref{model:1} which satisfy the conditions $|x_i^1-x_i^2|=1$, and
the boundary conditions. $(a)$ A solution that corresponds to the movement
along half the cycle in Fig. \ref{f1}; $(b)$ full cycle; $(c)$ two
full cycles; $(d)$ four full cycles. Changing $d_i$ and hence the
velocity of the movement along the phase curves we can always obtain
solutions with arbitrary number of full cycles}\label{f2}
\end{figure}

\paragraph{Remark 4.2.} We introduce the following parameter
\begin{equation}\label{4:13}
    d=\sum_{i=1}^n\left(\frac{d_i}{a_i}\right)^\frac{1}{p}.
\end{equation}
Theorem \ref{th4:1} can be restated as follows: If
$$
d<\mu=(p\pi^{-2})^\frac{1}{p}
$$
then there exists a spatially non-uniform stationary solution to
\eqref{model:1}. On the other hand, if $d>\mu$ we obtain from
Theorem \ref{th3:1} that $R_i$ are stable. Therefore we can consider
$d$ as a bifurcation parameter. As this parameter decreases
spatially uniform stationary solutions become unstable, and
spatially non-uniform solutions appear in the system according to
the standard Turing bifurcation scenario.\vspace{2mm}

Now we consider the case of the spatially explicit hypercycle
\eqref{model:2}.

\begin{theorem}\label{th4:2}
Suppose that \eqref{4:1} holds. If the parameters of problem
\eqref{model:2} can be represented by one-parameter perturbation
$$
d_i=d_0+\varepsilon l_i,\quad a_i=a_0+\varepsilon m_i,\quad
m_i,\,l_i\mbox{ are constant,}\quad \varepsilon>0,
$$
where $\varepsilon$ is a small parameter, then there exist spatially
non-uniform stationary solutions to system \eqref{model:2}.
\end{theorem}
\begin{proof}
System \eqref{model:2} can be rewritten in the following form:
\begin{equation}\label{4:14}
    d_iu_i''+u_i(a_iu_i-\bar{f}_1)=u_ia_i(u_i-u_{i-1})+(\bar{f}_1-\bar{f}_2)u_i,\quad
    i=1,\ldots,n.
\end{equation}

If $\varepsilon=0$ then we have that $\bar{f}_1=\bar{f}_2$ and
\begin{equation}\label{4:15}
    d_0u_i''+u_i(a_0u_i-\bar{f}_1)=0,
\end{equation}
which is a particular case of the autocatalytic system
\eqref{model:1:1}. According to Theorem \ref{th4:1} system
\eqref{4:15} possesses spatially non-uniform stationary solutions.
Using the presentation \eqref{4:4} it can be shown that the
right hand side of \eqref{4:14} is of the order of $\varepsilon$,
i.e., can be rewritten in the form
\begin{equation}\label{4:16}
    d_iu_i''+u_i(a_iu_i-\bar{f}_1)=\varepsilon \Psi_i(x),\quad
    i=1,\ldots,n.
\end{equation}
where $\Psi_i(x)$ are bounded functions. This implies that system
\eqref{4:16} is a perturbation of the Hamiltonian system
\eqref{4:15}. According to the general
theory \cite{guckenheimer1983nod} stable and unstable manifolds of
the perturbed orbits will be close to the corresponding manifolds of
the unperturbed system. Therefore for $d<\mu$ for each non-uniform
stationary solution of \eqref{model:1} there exists spatially
non-uniform stationary solution to \eqref{model:2}.
\end{proof}

\paragraph{Remark 4.3.} If we assume that the inverse to inequality
\eqref{3:10} holds, then it can be rewritten in the form
$$
\sum_{i=1}^n\left[\frac{d_i}{a_i}\right]^\frac{1}{p}>\frac{p^\frac{1}{p}}{\pi^\frac{2}{p}}.
$$
Indeed, we can rewrite inverse to \eqref{3:10} in the form
$$
\left[\prod_{i=1}^n
d_i\right]^\frac{1}{pn}>\frac{p^\frac{1}{p}}{\beta\pi^\frac{2}{p}}.
$$
Using the properties of arithmetic and geometric means we obtain
$$
\beta=\sum_{i=1}^n\frac{1}{(a_i)^{\frac{1}{p}}}\leq
n\left[\prod_{i=1}^n\frac{1}{(a_i)^\frac{1}{p}}\right]^\frac{1}{n}.
$$
From the previous it follows that
$$
n\left[\prod_{i=1}^n
d_i\right]^\frac{1}{pn}\left[\prod_{i=1}^n\frac{1}{(a_i)^\frac{1}{p}}\right]^\frac{1}{n}=n\left[\prod_{i=1}^n\left[\frac{d_i}{a_i}\right]^\frac{1}{p}\right]^\frac{1}{n}>\frac{p^\frac{1}{p}}{\pi^\frac{2}{p}}.
$$
Once again using the inequality between arithmetic and geometric
means we obtain
$$
n\left[\prod_{i=1}^n\left[\frac{d_i}{a_i}\right]^\frac{1}{p}\right]^\frac{1}{n}\leq
\sum_{i=1}^n\left[\frac{d_i}{a_i}\right]^\frac{1}{p},
$$
which proves the desired result.

In words, we showed in this remark that if the inverse to \eqref{3:10} holds, then  the inverse to \eqref{4:1} is true, which means that if the spatially homogeneous solution to hypercycle system is stable there are no spatially non-homogeneous solutions.

\paragraph{Example 4.1.} It is possible to obtain an explicit
solution to \eqref{model:1:2} in the special case when $n=4$,
$d_1=d_3,\,d_2=d_4,\,a_i=a,\,i=1,2,3,4$. First, we rewrite
\eqref{inv:2:1} in the form
$$
\bar{f}_2=\int_0^1 \langle Au, u\rangle dx,\quad
u=(u_1,u_2,u_3,u_4)'.
$$
The expression $\langle u,v\rangle$ denotes the standard scalar product
in $\mathbb R^4$, $'$ is the transformation. Matrix $A$ is circular
and has eigenvalues
$\lambda_1=a,\,\lambda_3=-a,\,\lambda_2=\lambda_4=0$. Consider the
orthogonal transformation that reduces $A$ to its canonical form:
$$
T=\left(%
\begin{array}{cccc}
  1/2 & \sqrt{2}/2 & -1/2 & 0 \\
  1/2 & 0 & 1/2 & -\sqrt{2}/2 \\
  1/2 & \sqrt{2}/2 & -1/2 & 0 \\
  1/2 & 0 & 1/2 & -\sqrt{2}/2 \\
\end{array}%
\right).
$$

Summing all equations in the hypercyclic system we have

$$
\langle u_{xx},D\rangle=\bar{f}_2\langle u,
\textbf{1}\rangle-\langle Au, u\rangle,
$$
where $D$ is the diffusion vector, $D=(d_1,d_2,d_1,d_2)$,
$\textbf{1}=(1,1,1,1)$. Let $u(x)=Tv(x)$, $v=(v_1,v_2,v_3,v_4)$. It
follows that
$$
\langle v_{xx},T'D\rangle=\bar{f}_2\langle v,
T'\textbf{1}\rangle-\langle T'ATV, v\rangle.
$$
Since $\langle T'ATV, v\rangle=k(v_1^2-v_3^2)$, the last equation
takes the form
$$
(d_1+d_2)(v_1)_{xx}=2\bar{f}_2v_1-k(v_1-v_3).
$$
Suppose that $u_1+u_3=u_2+u_4$. Then we have that $v_3=0$ and the
function $w(x)=u_1+u_2$, satisfies the differential equation
$$
(d_1+d_2)w_{xx}=2\bar{f}_2w-aw^2,
$$
whose explicit solution can be found using \eqref{4:4}.

\paragraph{Remark 4.4.} As in the case of system \eqref{model:1}
parameter $d$ can be considered as a bifurcation parameter for
\eqref{model:2}.

\section{Asymptotic behavior of the spatially explicit autocatalytic and hypercyclic systems}
Consider the local system of autocatalytic reaction \eqref{intro:0}
in the form
\begin{equation}\label{5:1}
\begin{split}
    \frac{dw_i}{dt}&=w_i(a_iw_i^p-f_1^{loc}(t)),\quad
    f_1^{loc}(t)=\sum_{i=1}^na_iw_i^{p+1},\quad t>s,\\
    w_i(s)&=\xi_i,\quad i=1,\ldots,n, \quad \sum_{i=1}^nw_i=1.
\end{split}
\end{equation}

For the following we need
\begin{definition}
We shall say that the initial conditions for system
\eqref{model:1} and system \eqref{5:1} are concerted if

\begin{equation}\label{5:2}
\xi_i=\bar{\varphi}_i=\int_{\Omega}\varphi_i(x)dx.
\end{equation}

\end{definition}

Let us assume that the initial conditions for systems
\eqref{model:1} and \eqref{5:1} are concerted. On integrating system
\eqref{model:1} with respect to $x$ and using the equality
$\int_{\Omega}\Delta v(x)dx=\int_{\Gamma}\partial_{\textbf{n}}v
ds=0$ we obtain
$$
\frac{d\bar{v}_i}{dt}=a_i\int_{\Omega}v_i^{p+1}(x,t)dx-\bar{v}_i(t)f_1^{loc}(t),\quad
t>s, \quad \bar{v}_i(s)=\bar{\varphi}_i=\xi_i,\quad i=1,\ldots, n,
$$
where
$$
\bar{v}_i(t)=\int_{\Omega}v_i(x,t)dx.
$$
Since $|\Omega|=1$ we have
$$
\int_\Omega
v_i^{p+1}(x,t)dx\geq\left(\int_{\Omega}v_i(x,t)\,dx\right)^{p+1}=\bar{v}_i^{p+1}(t),
$$
and, consequently,

\begin{equation}\label{5:3}
\frac{d\bar{v}_i}{dt}\geq \bar{v}_i^{p+1}(t)
-\bar{v}_i(t)f_1^{loc}(t),\quad t>s, \quad
\bar{v}_i(s)=\bar{\varphi}_i=\xi_i,\quad i=1,\ldots, n,
\end{equation}

\begin{lemma}\label{l5:1}
Let the initial conditions for systems \eqref{model:1} and
\eqref{5:1} be concerted. Then
\begin{equation}\label{5:4}
    \beta^{-p}\leq f_1^{loc}(t)\leq f_1(t),
\end{equation}
where $\beta=\sum_{i=1}^na_i^{-\frac{1}{p}}$.
\end{lemma}

\begin{proof}
First, we prove the left inequality in \eqref{5:4}. Using \eqref{2:1}
and H\"{o}lder's inequality we have
$$
1=\left(\sum_{i=1}^nw_i(t)\right)^{p+1}\leq\left[\sum_{i=1}^n\frac{1}{a_i^p}\right]^p\sum_{i=1}^na_iw_i^{p+1}=\beta^pf_1^{loc}(t).
$$

To prove the right inequality in \eqref{5:4} we assume that there exists $s\geq 0$ that $f_1(s)< f_1^{loc}(s)$. Since the functions $f_1(t)$ and $f_1^{loc}$ are continuous, there exists neighborhood $U_{\delta}=\{t\colon 0\leq t-s<\delta\}$ from which $f_1(t)<f_1^{loc}(t)$ follows. Then from \eqref{5:3} it follows that
\begin{equation}\label{5:5}
    \frac{d\bar{v}_i}{dt}\geq \bar{v}_i^{p+1}(t)-\bar{v}_i(t)f_1(t),\quad t\in U_{\delta},\quad i=1,\ldots,n.
\end{equation}
Due to the fact that the initial conditions of \eqref{model:1} and \eqref{5:1} are concerted, then from the comparison theorem \cite{tikhonov1985de} we obtain
\begin{equation}\label{5:6}
    \bar{v}_i(t)>w_i(t),\quad t\in U_{\delta},\quad i=1,\ldots,n,
\end{equation}
where $w_i(t)$ are the solutions to \eqref{5:1}. From the other hand we should have $\sum_{i=1}^n\bar{v}_i(t)=\sum_{i=1}^nw_i(t)=1$; we obtain a contradiction.

\end{proof}

\begin{theorem}\label{th5:1}
Let $p\geq 1$. Then for almost all initial conditions
$\varphi_i(x),\,\sum_{i=1}^n\int_{\Omega}\varphi_i(x)dx=1$ there
exists an index $j$, $1\leq j\leq n$ \emph{(}which depends
on $\varphi_i(x)$\emph{)} such that ${v}_i(x,t)\to0$ for all
$i\neq j$ in the space $L_{p+1}$, and $\int_\Omega{v}_j(x,t)\to 1$ when $t\to\infty$.
\end{theorem}
\begin{proof}
We have $p\geq 1$, and hence $H_{p+1}^1 \subseteq H_2^1$. The
eigenfunctions $\psi_s,\,s=0,1,2,\ldots$ of the problem
\eqref{eigenv:pr} form a complete system in $H_2^1$. Let us
represent
$$
\varphi_i(x)=\bar{c}_i^0+z_i(x),\quad z_i(x)=\sum_{s=1}^\infty
\bar{c}_i^s\psi_s(x).
$$
Let $w_i(t)$ be the solutions to \eqref{5:1} and let the initial
conditions for systems \eqref{model:1} and \eqref{5:1} be
concerted. We will look for a solution to \eqref{model:1} in the
form
\begin{equation}\label{5:7}
    \begin{split}
    v_i(x,t) & =w_i(t)+z_i(x,t),\quad z_i(x,t)=\sum_{s=0}^\infty c_i^s(t)\psi_s(x),\\
       w_i(0) & =c_i^0,\quad c_i^m(0)=\bar{c}_i^m,\quad m=1,2,\ldots
\end{split}
\end{equation}
Inserting
\eqref{5:7} into \eqref{model:1} we obtain
$$
\frac{dw_i(t)}{dt}+\frac{\partial z_i(x,t)}{\partial
t}=a_iv_i^{p+1}(x,t)-f_1(t)(w_i(t)+z_i(x,t))+d_i\frac{\partial^2
z_i(x,t)}{\partial x^2}.
$$
Integrating the last equation with respect to $x$ and noting
$\int_\Omega \psi_s(x)dx=0$ give
$$
\frac{dw_i(t)}{dt}=a_i\int_\Omega v_i^{p+1}(x,t)dx-f_1(t)w_i(t).
$$
Using the fact that $w_i(t)$ are the solutions to \eqref{5:1} we
obtain
\begin{equation}\label{5:8}
    a_i\int_\Omega
    v_i^{p+1}(x,t)dx=(f_1(t)-f_1^{loc}(t))w_i(t)+a_iw_i^{p+1}(t).
\end{equation}
It is known that solutions to \eqref{5:1} have a property of
multistability. It means that all the vertexes of the simplex are
stable, and the choice of initial conditions determines to which
vertex the system evolves. In other words, for almost all initial
conditions $\xi_i$ the system \eqref{5:1} ends up in $R_j$, for
which all the coordinates excluding $j$ are zero ($w_i(t)\to0$ when
$t\to\infty$ for all $i\neq j$, and $w_j(t)\to 1$). Hence, from
\eqref{5:8} the theorem follows.
\end{proof}

\paragraph{Remark 5.1.} Theorem \eqref{th5:1} answers a natural
question which spatially non-uniform stationary solution of
\eqref{model:1} survives in the evolutionary process. To answer it we
need to consider two systems \eqref{model:1} and \eqref{5:1} with
concerted initial conditions. As was mentioned system \eqref{5:1}
possesses the property of multistability; each vertex of the simplex
has its own basin of attraction. If we denote these basins as
$D_1,\ldots,D_n$, then the number of the basin, to which the initial
conditions of \eqref{5:1} belong, determines which spatially
non-uniform solution will dominate the evolution. Note that for the
dominant solution
\begin{equation}\label{5:9}
    \int_{\Omega} v(x,t)dx\to 1\mbox{ for } t\to\infty.
\end{equation}

Another point here is that the explicit space structure in the system
with global regulation \eqref{model:1} does not provide the
conditions for surviving more than one type of prebiotic replicators, $v_i(x,t)\to 0$ in $L_{p+1}$ for all $i\neq j$.

\begin{corollary}
Almost all spatially non-uniform stationary solutions of the problem
\eqref{model:1} are unstable.
\end{corollary}
\begin{proof}
Consider \eqref{model:1} with the initial conditions
$$
v_i(x,s)=u_i(x)+\beta_i(x),\quad i=1,\ldots, n,
$$
where $u_i(x)$ are spatially non-uniform stationary solutions to
\eqref{model:1}, and $\beta_i(x)\in H_{\delta}$. From Theorem
\ref{th5:1} it follows that there exists a positive integer $j$
(which depends on the initial conditions) such that $v_i(x,t)\to 0$ in space $L_{p+1}$
for $i\neq j$. Therefore only one set of stationary solutions can be
stable.
\end{proof}

\paragraph{Remark 5.2.} It is possible to obtain sufficient
conditions for stability of the non-uniform stationary solution
$u_j(x)>0,\,\int_{\Omega}u_j(x)dx=1$. Unfortunately, applying this
condition requires additional serious analysis.

Indeed, we can look for a solution to \eqref{model:1} in the form
$$
v_j(x,t)=u_j(t)+z_j(x,t),\quad v_i(x,t)=z_i(x,t),\,i\neq j.
$$
Putting these solutions into \eqref{model:1} and retaining only
linear terms we obtain
\begin{equation*}
\begin{split}
\partial_t
z_j(x,t) &=a_j(p+1)[u_j^p(x)z_j(x,t)-\langle
u_j^p(x),z_j(x,t)\rangle]-a_jz_j(x,t)\langle
u_j^{p+1}(x),1\rangle+d_i\Delta
z_j(x,t),\\
\partial_t
z_i(x,t) &= -a_iz_i(x,t)\langle u_i^{p+1}(x),1\rangle+d_i\Delta
z_i(x,t),\quad i\neq j,
\end{split}
\end{equation*}
with the initial conditions $z_i(x,s)=u_i(x)+\beta_i(x)$. Here
$\langle u(x),v(x)\rangle$ denotes the usual scalar product in
$L_2(\Omega)$. This implies that all $z_i(x,t)\to 0$ for $i\neq j$
when $t\to\infty$. On the other hand we have
\begin{equation}\label{add}
    \begin{split}
    \frac{1}{2}\frac{d}{dt} \int_{\Omega}z_j^2(x,t)dx& =a_j(p+1)[\langle u_j^p(x), z_j^2(x,t)\rangle-\langle u_j(x),z_j(x,t)\rangle\langle u_j^p(x),z_j(x,t)\rangle]\\
        &{ } -a_j\langle z_j(x,t),z_j(x,t)\rangle \langle u_j^{p+1},
        1\rangle+d_i\langle \Delta z_j(x,t), z_j(x,t)\rangle.
\end{split}
\end{equation}
Substituting the following
$$
z_j(x,t)=z_j^0(t)+\sum_{s=1}^\infty z_j^s(t)\psi_s(x)
$$
into \eqref{add} and using the fact that
$$
\langle \Delta z_j(x,t),z_j(x,t)\rangle=-\sum_{s=1}^\infty
\lambda_s(z_j^s(t))^2
$$
we obtain that all the terms in \eqref{add} except for the terms in
the square brackets are negative. The terms in the square brackets have
the following form
$$
\alpha=(p+1)\sum_{m=1}^\infty\sum_{s=1}^\infty
a_jz_j^m(t)z_j^s(t)(\langle u_j^p(x),
\psi_m(x)\psi_s(x)\rangle-\langle u_j(x),\psi_m(x)\rangle \langle
u_j^p(x),\psi_s(x)\rangle),
$$
from which we obtain a sufficient condition for stability of the
solution $u_j(x)>0$ in the form
$$
\alpha<\sum_{m=1}^\infty(z_j^m(t))^2(\lambda_m+a_j\bar{u}_j^p).
$$
The last formula should be checked only for small $m$ because $\lambda_m\to\infty$.

The result of Lemma \ref{l5:1} can be extended to the case of hypercycle reaction.
\begin{lemma}\label{l5:1n}
Let the initial conditions of system \eqref{model:2} and system
\begin{equation}\label{5:n:1}
\begin{split}
    \frac{dw_i}{dt}&=w_i(a_iw_{i-1}^p-f_2^{loc}(t)),\quad
    f_2^{loc}(t)=\sum_{i=1}^na_iw_iw_{i-1}^{p},\quad t>0,\\
    w_i(0)&=\xi_i,\quad i=1,\ldots,n, \quad \sum_{i=1}^nw_i=1
\end{split}
\end{equation}
be concerted. Then
\begin{equation}\label{5:n:2}
    f_2^{loc}(t)\leq f_2(t).
\end{equation}
\end{lemma}
\begin{proof}
We have
$$
\int_{\Omega}\frac{\Delta v_i}{v_i}\,dx=\int_{\Gamma}\frac{\partial v_i}{\partial \textbf{n}}\frac{1}{v_i}\,ds+\int_{\Omega}|\nabla_x\ln v_i|^2dx\geq 0,
$$
and
$$
\int_{\Omega}v^p_{i-1}(x,t)dx\geq \left(\int_\Omega v_{i-1}(x,t)dx\right)^p=\bar{v}_{i-1}^p.
$$
Therefore
$$
\int_{\Omega}\frac{\partial}{\partial t}\ln v_i\,dx\geq a_i\bar{v}^p_{i-1}(t)-f_2(t).
$$
Since the initial conditions of \eqref{model:2} and \eqref{5:n:1} are concerted, then, as in the case of Theorem \eqref{th5:1} we can represent $v_i(x,t)$ as the sum $v_i(x,t)=w_i(t)+z_i(x,t)$, where $z_i(x,t)$ are given by \eqref{5:7}, and note that $w_i(0)=\bar{v}_i(0)$ for any $i$.

From the last inequality it follows that
$$
\int_{\Omega}\frac{\partial}{\partial t}\ln v_i \, dx=\frac{d}{dt}w_i(t)\geq a_i\bar{v}_{i-1}^p-f_2(t).
$$
Since $\bar{v}_i(0)=w_i(0)$ then, using \eqref{5:n:1} we obtain \eqref{5:n:2}.
\end{proof}

Using the last lemma we can extend the results of \textit{permanence} of hypercycle system with $p=1$ to the spatially explicit case \cite{hofbauer1998ega}. We remind that permanence means that solutions to system \eqref{5:n:1} with the initial conditions $w_i(0)=\xi_i>0$ do not vanish, i.e.,
$$
1>w_i(t)>\delta>0, \quad t>0.
$$
\begin{corollary}Let $p=1$ and let the initial conditions of systems \eqref{model:2} and \eqref{5:n:1} be concerted, and
$$
\bar{\varphi_i}=\xi_i=w_i(0)>0,\quad \sum\xi_i=1.
$$
Then the solutions to system \eqref{model:2} do not vanish in $L_2$ space.
\end{corollary}
\begin{proof}
Let a solution $v_i(x,t)$ to \eqref{model:2} vanish for some $i$, i.e.,
$$
\|v_i(x,t)\|_{L_2}\to 0, \quad t\to\infty.
$$
Using the reasoning along the lines of Theorem \ref{th5:1}, we obtain
$$
a_i\int_{\Omega}v_i(x,t)v_{i-1}(x,t)\,dx= (f_2(t)-f_2^{loc}(t))w_i(t)+a_iw_i(t)w_{i-1}(t).
$$
Using Lemma \ref{l5:1n} we hence have
$$
\int_{\Omega}v_i(x,t)v_{i-1}(x,t)\,dx\geq w_i(t)w_{i-1}(t).
$$
The last and the Cauchy inequalities yield
$$
\|v_i(x,t)\|_{L_2}\|v_{i-1}(x,t)\|_{L_2}\geq w_i(t)w_{i-1}.
$$
From the fact $\|v_i(x,t)\|_{L_2}\to 0$ it follows that either $w_i$ or $w_{i-1}$ tend to zero, which contradicts to the permanence of the hypercycle system \eqref{5:n:1}. This completes the proof.
\end{proof}

Similar to Remark 5.2 we can obtain sufficient conditions for stability of the spatially nonhomogeneous stationary solutions for the hypercycle system \eqref{model:2}. However, the utility of such conditions is questionable because we hardly can expect that we will be able to check these conditions analytically.

It is possible to study the stability of spatially nonhomogeneous solutions in somewhat weaker sense.

\begin{definition}
We shall say that spatially non-uniform stationary solution
$u(x)=(u_1(x),\ldots,u_n(x))$ to system \eqref{model:1} or
\eqref{model:2} is stable in the sense of the mean integral value
if for any $\varepsilon>0$ there exists $\delta>0$ such that for
the initial conditions
$$
|\bar{\varphi}_i-\bar{u}_i|<\delta,
$$
it follows that
$$
|\bar{v}_i-\bar{u}_i|<\varepsilon,
$$
for any $i$ and $t>0$, where, as before, $v_i(x,t)$ are the
solutions of \eqref{model:1} or \eqref{model:2},
$$
\bar{v}_i=\int_\Omega v_i(x,t)dx,\quad\bar{\varphi}_i=\int_\Omega
\varphi_i(x)dx,\quad \bar{u}_i=\int_\Omega u_i(x)dx.
$$
\end{definition}

It is clear that the stability in the mean integral sense is weaker
than the stability in the usual sense (Lyapunov stability). For example,
consider functions $g(x,t)\in H_2^1,\,x\in[0,1]$
$$g(x,t)=c_0(t)+\sum_{s=1}^\infty c_k(t)\cos k\pi x.$$
Let us suppose that $c_0(t)\to 0$ when $t\to\infty$. Then
$\bar{g}(t)\to0$ whereas $\|g(x,t)\|^2_{H_2^1}=\sum_{s=1}^\infty
c_k^2(t)(1+k^2\pi^2)$ does not necessarily tend to zero.

\begin{corollary}
Let us suppose that the following inequalities hold for any
$i=1,\ldots,n$:
\begin{equation}\label{5:10}
    \frac{d_i}{a_i}<\frac{p}{\lambda_1}\,.
\end{equation}
Then all spatially non-uniform stationary solutions to
\eqref{model:1} of the form $$U^j(x)=(0,\ldots,0,u_j(x),0,\ldots,0)$$
are stable in the mean integral sense.
\end{corollary}
\begin{proof}
From Theorem \ref{th3:1} it follows that $R_j$ are unstable when
\eqref{5:10} holds. Consider the solution $U^j(x)$ for which
$\bar{u}_j=1$. From the other hand from Theorem \eqref{th5:1}
follows \eqref{5:9}, which completes the proof.
\end{proof}

Now we switch to the hypercycle system \eqref{model:2} with explicit
spatial structure and global regulation. After integrating
\eqref{model:2} with respect to spatial variable, we obtain
\begin{equation}\label{5:11}
    \frac{\bar{v}_i(t)}{dt}=a_i\langle
    v_{i-1}^p,v_i\rangle-f_2(t)\bar{v}_{i}(t),\quad 0<t,\quad
    v_i(0)=\bar{\varphi}_i,
\end{equation}
where the meaning of the function $\bar{g}$ as before, in the mean
integral sense.

Let us introduce new functions
\begin{equation}\label{5:12}
    v_i(x,t)=\frac{w_i(x,t)}{(a_i)^\frac{1}{p}}\Theta,\quad
    i=1,\ldots, n,\quad
    \Theta=\int_{\Omega}\sum_{j=1}^{n}(a_j)^\frac{1}{p}v_j(x,t)dx.
\end{equation}
For the new variables
\begin{equation}\label{5:13}
    \sum_{j=1}^{n}\int_\Omega w_j(x,t)dx=1.
\end{equation}
Note that in the new variables the stationary point $P_2$ has the
coordinates $(1/n,1/n,\ldots,1/n)$.

\begin{lemma}\label{l5:2}
In the new variables \eqref{5:12} the dynamical system \eqref{5:11}
has the following form:
\begin{equation}\label{5:14}
\begin{split}
    \frac{d\bar{w}_i(t)}{dt} & =\Theta^p(t)(\langle w_{i-1}^p,w_i \rangle-f_2(t)\bar{w}_i(t)),\\
       f_2(t) & = \sum_{j=1}^{n}\int_\Omega w_{j}(x,t)w_{j-1}^p(x,t)dx.
\end{split}
\end{equation}
\end{lemma}
\begin{proof}
Using \eqref{5:12}, \eqref{model:2} and the boundary conditions we
obtain

\begin{equation*}
\begin{split}
\frac{d\Theta}{dt}&=\sum_{j=1}^{n}\int_\Omega(a_jv_jv_{j-1}^p-f_2(t)v_j+d_j\Delta
v_j)(a_{j+1})^\frac{1}{p}dx\\&=\Theta^{p+1}\sum_{j=1}^{n}\int_\Omega
 w_j w_{j-1}^p
dx-f_2(t)\Theta\sum_{j=1}^{n}\int_\Omega w_jdx
\end{split}
\end{equation*}
Using \eqref{5:13} we obtain
\begin{equation}\label{5:15}
\frac{d\Theta}{dt}=\Theta^{p+1}\sum_{j=1}^{n}\int_\Omega
w_jw_{j-1}^pdx-f_2(t)\Theta.
\end{equation}
Equality \eqref{5:12} yields
\begin{equation}\label{5:16}
    \frac{d\bar{v}_i}{dt}=\frac{\dot{\bar{w}}_i\Theta+\bar{w}_i\dot{\Theta}}{(k_{i+1})^\frac{1}{p}}\,.
\end{equation}
From the other hand \eqref{5:12} implies
\begin{equation}\label{5:17}
    \frac{d\bar{v}_i}{dt}=\frac{\Theta^{p+1}\langle w_i,w_{i-1}^p\rangle-\Theta \bar{w}_if_2(t)}{(k_{i+1})^\frac{1}{p}}\,.
\end{equation}
Putting together \eqref{5:15},\eqref{5:16} and \eqref{5:17}
completes the proof.
\end{proof}

Consider the spatially uniform stationary solution $P_2$ to \eqref{model:2}.
It is also an equilibrium of \eqref{5:11}.

\begin{theorem}
Let $u(x)=(u_1(x),\ldots,u_n(x))$ be a spatially non-uniform
stationary solution to \eqref{model:2} such that $\bar{u}=P_2$,
where $P_2$ is the homogeneous stationary solution of
\eqref{model:2}. Then $u(x)$ is stable in the sense of the mean
integral value.
\end{theorem}
\begin{proof}
Consider system \eqref{5:11}. Due to Lemma \ref{l5:2} this system is
topologically equivalent to system \eqref{5:14}, which has the
steady state $P_0=(1/n,\ldots,1/n)$. Let us introduce the following
Laypunov function
$$
V(\bar{w}_1,\ldots,\bar{w}_n)=-\ln(\bar{w}_1\bar{w}_2\ldots\bar{w}_n)-n\ln
n.
$$
It is easy to see that $V(P_0)=0$ and
$V(\bar{w}_1,\ldots,\bar{w}_n)>0$ in a neighborhood $Z_\delta$ of
$P_0$, where
$$
Z_\delta=\left\{\bar{w}_i,i=1,\dots,n\,\,:\,\,\sum_{j=1}^n\bar{w}_j=1,\quad\sum_{j=1}^n|\bar{w}_j-\frac{1}{n}|\leq
\delta\right\}.
$$
Using \eqref{5:11} yields
\begin{equation*}
    \begin{split}
\dot{V}&=-\sum_{i=1}^n\frac{\dot{\bar{w}}_i}{\bar{w}_i}=\\
&=-\Theta^p\sum_{i=1}^n \left[\frac{\langle
w_i,w_{i-1}^p\rangle}{\bar{w}_i}-f_2(t)\right]=\\
&=-\Theta^p\sum_{i=1}^n \langle
w_i,w_{i-1}^p\rangle(\frac{1}{\bar{w}_i}-n).
\end{split}
\end{equation*}

Denote $\mu$ the following
$$
\mu=\min_{1\leq i\leq n}\left\{\inf_t \langle
w_i,w_{i-1}^p\rangle\right\}.
$$
The functions $w_i(x,t)$ are nonnegative for all $i$, therefore we
obtain
$$
\dot{V} \leq
-\Theta^p\mu\left(\sum_{i=1}^n\frac{1}{\bar{w}_i}-n^2\right).
$$

We also have $\sum_{i=1}^n\frac{1}{\bar{w}_i}\geq
\frac{n}{\sqrt[n]{\prod_{i=1}^n\bar{w}_i}}$. Since
$\sum_{i=1}^n\bar{w}_i=1,\,\bar{w}_i\geq 0$, the function
$\prod_{i=1}^n\bar{w}_i$ reaches its maximum at the point
$P_0=(1/n,\ldots,1/n)$, and this implies that
$$
\frac{n}{\sqrt[n]{\prod_{i=1}^n\bar{w}_i}}\geq n^2
$$
which means that $\dot{V}\leq 0$. Invoking the arguments of the topological equivalence of \eqref{5:11} and \eqref{5:14} completes the proof.
\end{proof}
\section{Conclusion}
In this paper we studied the existence and stability of stationary solutions to autocatalytic and hypercyclic systems \eqref{model:1} and \eqref{model:2} with nonlinear growth rates and explicit spatial structure. It is well known that the mean field models (e.g., models described by ODE systems) are often show different behavior from the models where the spatial structure is taken into consideration (more on this \cite{dieckmann2000}). In particular, it is widely acknowledged that the evolution and survival of altruistic traits can be mediated by spatial heterogeneity. Macromolecules that catalyze the production of other macromolecules are obviously altruists, and in this note we tried to answer the question whether the particular form of spatial regulation (namely, global regulation \cite{bratus2006ssc,weinberger1991ssa}) can promote the coexistence of different types of macromolecules in the prebiotic world (within a hydrothermally formed system of continuous iron-sulfide compartments \cite{Koonin2005}). The analysis presented in \cite{bratus2006ssc,weinberger1991ssa} is significantly extended to the cases of nonlinear growth rates, arbitrary fitness and diffusion coefficients.

The major conclusion is as follows: the mathematical models with spatial structure and global regulation show in general very similar qualitative features to those of local models. Two basic properties, namely the competitive exclusion for autocatalytic systems and the permanence for the hypercyclic systems, are shown to hold for spatially explicit systems. Numerical calculations illustrate these conclusions in Figs. \ref{f3} and \ref{f4} (the details on the numerical scheme used in the calculations are given in \cite{bratus2006ssc}).
\begin{figure}[!tb]
\centering
\includegraphics{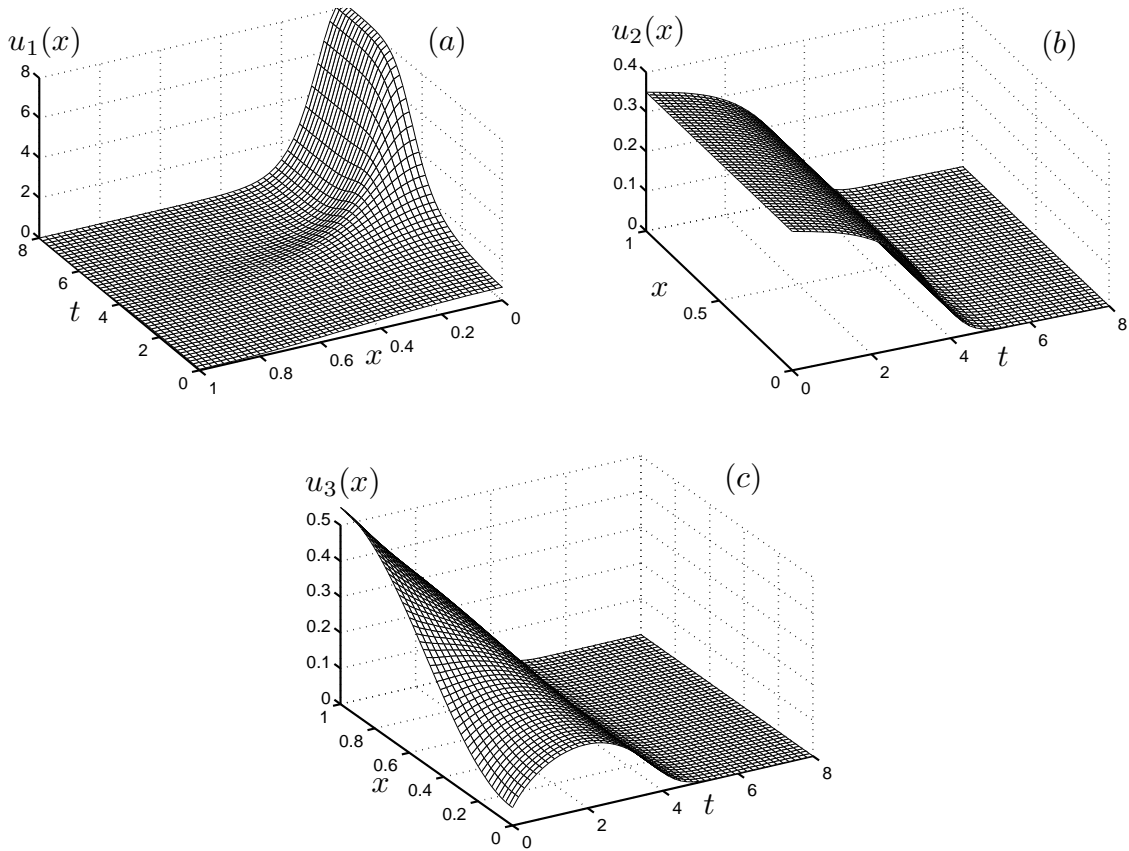}
\caption{The competitive exclusion for autocatalytic growth. Numerical solutions to autocatalytic system \eqref{model:1}. $n=3,\,d_1=0.02,\,d_2=0.05,\,d_3=0.08,\,p=1,\,a_1=a_2=a_3=1$. The initial conditions are $u_1(x,0)=0.35+0.3\cos \pi x,\,u_2(x,0)=0.35,\,u_3(x,0)=0.3-0.25\cos \pi x$. Note that the orientation of the axis is different for $(a)$ and $(b),\,(c)$. Only one type, $u_1$, survives. The asymptotic state is a spatially non-uniform stationary solution. The details of the numerical computations are given in \cite{bratus2006ssc}}\label{f3}
\end{figure}

\begin{figure}[!tb]
\includegraphics{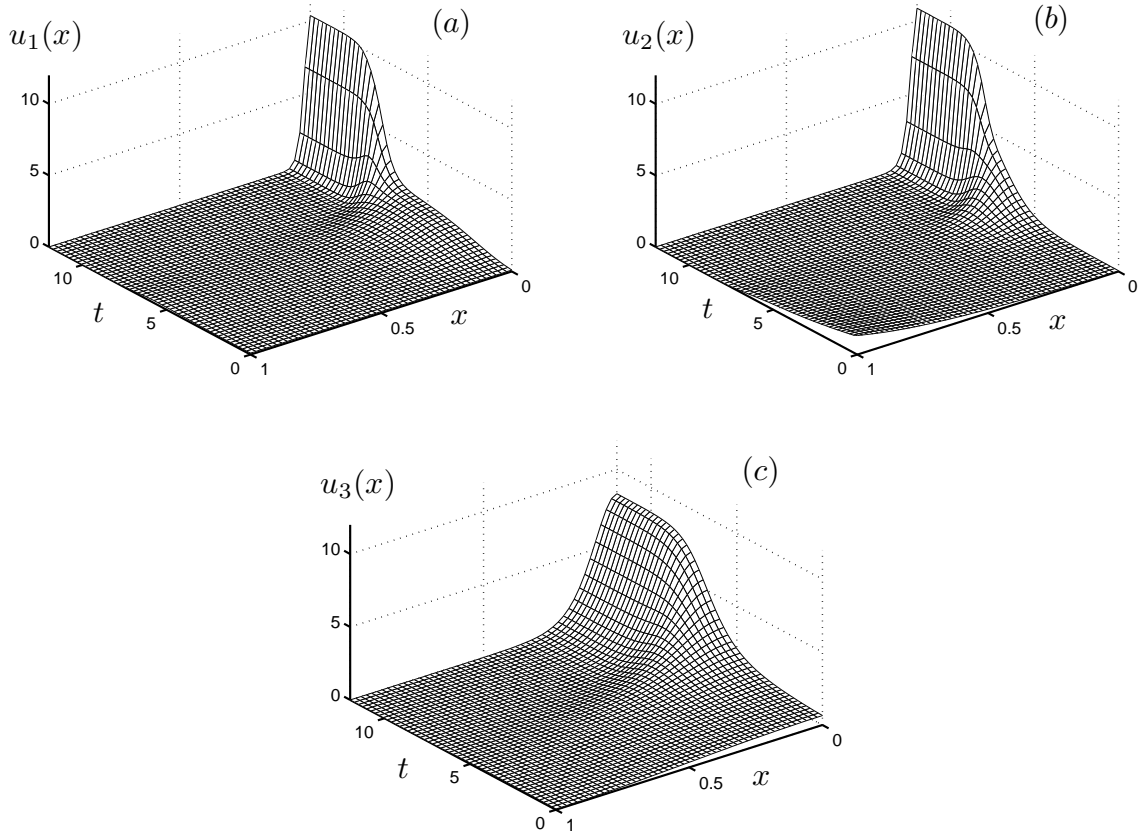}
\caption{The permanence for the hypercycle system. Numerical solutions to hypercyclic system \eqref{model:2}. $n=3,\,d_1=0.001,\,d_2=0.002,\,d_3=0.003,\,p=1,\,a_1=a_2=a_3=1$. The initial conditions are $u_1(x,0)=0.35+0.15\cos \pi x,\,u_2(x,0)=0.357,\,u_3(x,0)=0.338-0.3\cos \pi x$. The asymptotic state is spatially non-uniform stationary solutions. The details of the numerical computations are given in \cite{bratus2006ssc}}\label{f4}
\end{figure}

More precisely, for sufficiently large diffusion coefficients the spatially uniform stationary solutions to \eqref{model:1} and \eqref{model:2} have the same character as in the local models \eqref{intro:0} and \eqref{intro:1}. For such diffusion coefficients the asymptotic behavior of the local and distributed models coincides. If, on the other hand, the inequality \eqref{4:1} holds and the nonlinear growth rates satisfy the condition $0<p\leq 2$ then new, spatially non-uniform solutions appear; for small diffusion coefficients these spatially heterogeneous solutions can correspond to the multiple cycles on the phase plane of the corresponding Hamiltonian system (Fig. \ref{f2}). In the case of autocatalytic system these solution can be stable only if all but one asymptotic state are zero. In the case of the hypercyclic system we prove that these spatially heterogeneous solutions can be stable in the sense of the mean integral value. The examples of the asymptotic states for a hypercyclic systems found numerically are shown in Fig. \ref{f5}. These non-uniform stationary solutions can be considered as the means of the hypercycle system to withstand the parasite invasion \cite{smith1979hao} (the analysis of models with parasites and with $p>2$ is the subject of the ongoing work).

\begin{figure}[!tb]
\includegraphics{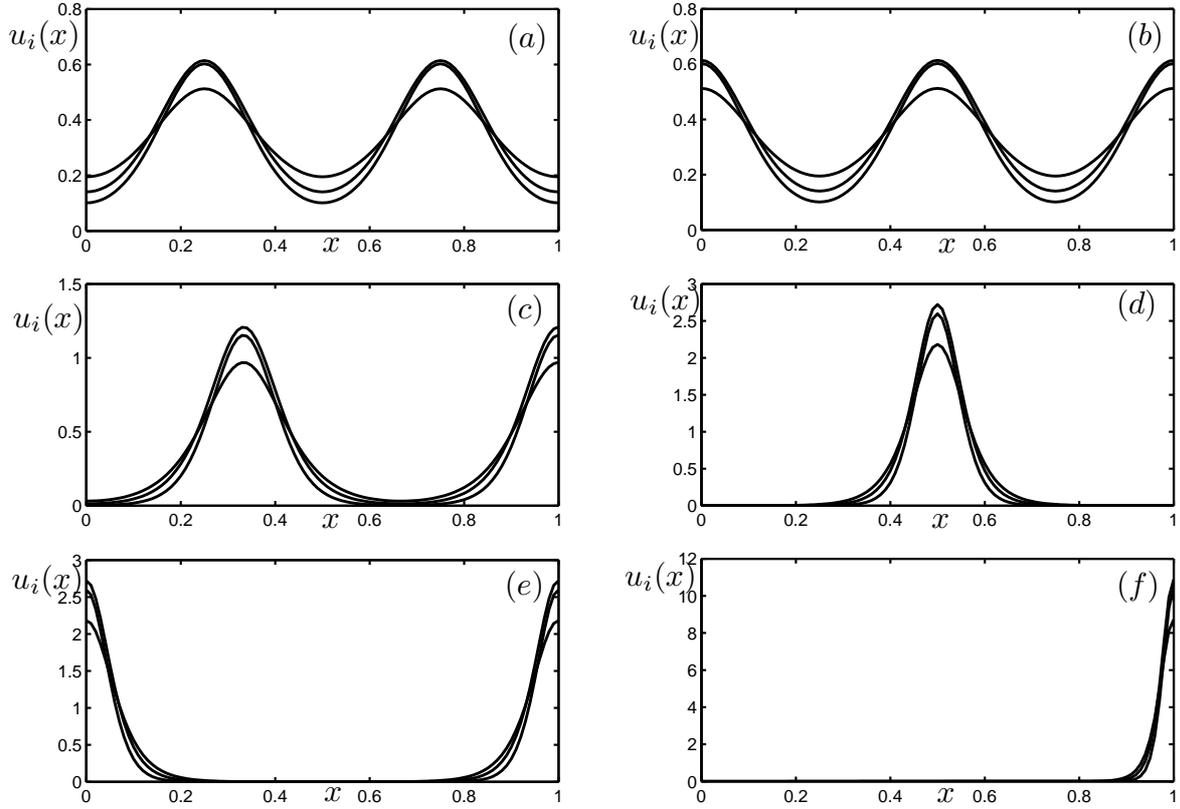}
\caption{Asymptotic spatially heterogeneous states of the hypercycle system \eqref{model:1:2} found numerically. $n=3,\,d_1=0.001,\,d_2=0.002,\,d_3=0.003,\,p=1,\,a_1=a_2=a_3=1$. Note that case $(f)$ corresponds to the simulation shown in Fig. \ref{f4}.}\label{f5}
\end{figure}

\appendix
\section{Appendix}
\begin{proof}[Proof of Lemma \ref{l4:1}] Consider the function
$g(t)=t^2-t^{p+2},\,p>0$. This function has two roots $\tau_1=0$ and
$\tau_2=1$, and attains its maximum at
$t^*=\left[\frac{2}{2+p}\right]^{\frac{2}{p}}$, which is
$g(t^*)=\left[\frac{2}{2+p}\right]^{\frac{2}{p}+1}$. Function
$P_i(t)$ can be obtained from $g(t)$ by shifting the latter.
Therefore, when \eqref{4:5} holds, $P_i(t)$ has two positive roots
that are situated in the interval \eqref{4:6}.
\end{proof}
\begin{proof}[Proof of Lemma \ref{l4:2}]
To simplify notations we drop indexes where it is possible. We need
to prove that for
\begin{equation*}
    P(\tau,c)=c+\tau^2-\tau^{p+2}
\end{equation*}
and
\begin{equation}\label{a:1}
I_1(c)=\int_{\tau_1}^{\tau_2}\frac{d\tau}{\sqrt{P(\tau,c)}},\quad
I_2(c)=\int_{\tau_1}^{\tau_2}\frac{\tau\,d\tau}{\sqrt{P(\tau,c)}},
\end{equation}
where
$$
\tau_1\in(0,\tau_0),\quad \tau_2\in(\tau_0,1),\quad
\tau_0=\left[\frac{2}{2+p}\right]^\frac{1}{p},\quad P(\tau_1,c)=0,\quad P(\tau_2,c)=0,\quad P'_{\tau}(\tau_0,c)=0,
$$
we have that
\begin{equation}\label{a:3}
    I(c)=\frac{1}{\tau_0}(I_1(c))^{\frac{2}{p}-1}I_2(c)\geq\left[\frac{\pi^2}{p}\right]^\frac{1}{p}
\end{equation}
for $0<p\leq2$.

For $p=2$ direct calculations show that $I_2(c)=\frac{\pi}{2},\,I(c)=\frac{\pi}{\sqrt{2}}$, hence
we assume that $0<p<2$. Using H\"{o}lder's inequality  yields
$$
(I_1)^{\frac{2}{p}-1}I_2=\left[(I_1)^{1-\frac{p}{2}}(I_2)^\frac{p}{2}\right]^\frac{2}{p}\geq
(I_3)^\frac{p}{2},
$$
where
$$
I_3(c)=\int_{\tau_1}^{\tau_2}\left(\frac{1}{\sqrt{P(\tau,c)}}\right)^{1-\frac{p}{2}}
\left(\frac{\tau}{\sqrt{P(\tau,c)}}\right)^{\frac{p}{2}}\,d\tau=\int_{\tau_1}^{\tau_2}\frac{\tau^\frac{p}{2}}{\sqrt{P(\tau,c)}}\,d\tau
$$

Next we will the following change of the variables:
$$
\tau^\frac{p+2}{2}=t,\quad \tau_1^\frac{p+2}{2}=t_1,\quad
\tau_2^\frac{p+2}{2}=t_2,\quad\frac{4}{p+2}=q,\quad
Q(t,c)=c+Q_0(t),\quad Q_0(t)=t^q-t^2,\quad 1<q<2,
$$
from which $Q(t_1,c)=Q(t_2,c)=0$, and hence $c=-Q_0(t_1)=-Q_0(t_2)$.

Out integral takes the form
$$
I_3(c)=\frac{q}{2}\int_{t_1}^{t_2}\frac{1}{\sqrt{Q(t,c)}}\,dt,
$$
$$
t_1\in(0,t_0),\quad t_2\in(t_0,1),\quad
t_0=\left(\frac{q}{2}\right)^\frac{1}{2-q},\quad Q_0'(t_0)=0.
$$

Function $Q(t,c)$ does not exceed its Hermite interpolation
polynomial $H_3$, which is build using the values at
$H_3(t_1)=H_3(t_2)=H_3'(t_0)=0,\,H_3(t_0)=Q(t_0,c)$. This follows from non-negativity of
the reminder term of interpolation
$$
Q(t,c)-H_3(t)=\frac{Q_0^{(4)}(\xi)}{24}(t-t_0)^2(t-t_1)(t-t_2),
$$
and the fact that $Q_0^{(4)}(\xi)>0$ when $t_1<\xi<t_2$. Therefore,
we have
$$
I_3(c)>\frac{q}{2}\int_{t_1}^{t_2}\frac{1}{\sqrt{H_3(t)}}\,dt,
$$
where
$$
H_3(t)=Q(t_0,c)\left(1-\frac{(t-t_0)(2t_0-t_1-t_2)}{(t_0-t_1)(t_0-t_2)}\right)\frac{(t-t_1)(t-t_2)}{(t_0-t_1)(t_0-t_2)}.
$$
Making the change of the variable in the integral
$$
t=\frac{t_1+t_2}{2}+\frac{t_2-t_1}{2}\sin\varphi,
$$
we obtain
$$
\int_{t_1}^{t_2}\frac{1}{\sqrt{H_3(t)}}\,dt=\sqrt{\frac{(t_0-t_1)(t_2-t_0)}{Q(t_0,c)}}I_4(c),
$$
where
$$
I_4(c)=\int_{-\pi/2}^{\pi/2}\frac{d\varphi}{\sqrt{1-\frac{((t_1+t_2)/2-t_0+(t_2-t_1)/2\sin\varphi)(2t_0-t_1-t_2)}{(t_0-t_1)(t_0-t_2)}}}.
$$
Since the graph of any convex function lays above any tangent line,
then we have
$$
\frac{1}{\sqrt{c_1x+c_2}}\geq\frac{1}{\sqrt{c_2}}\left(1-\frac{c_1x}{2c_2}\right)
$$
for any $x$. Using the last inequality we can estimate $I_4$ as
$$
I_4>
\int_{-\pi/2}^{\pi/2}\frac{1}{\sqrt{c_2}}\left(1-\frac{c_1\sin\varphi}{2c_2}\right)\,d\varphi=\frac{\pi}{\sqrt{c_2}}=\frac{\pi}{\sqrt{1-\frac{(2t_0-t_1-t_2)^2}{2(t_0-t_1)(t_2-t_0)}}}>\pi.
$$
Using the last estimate and returning to $I_3$ we obtain that
$$
I_3(c)>\frac{q\pi}{2}\sqrt{\frac{(t_0-t_1)(t_2-t_0)}{Q(t_0,c)}}=\frac{q\pi}{2}\sqrt[4]{g(t_1)g(t_2)},
$$
where
$$
g(t)=\frac{(t-t_0)^2}{Q_0(t_0)-Q_0(t)}\,.
$$

With the help of the Taylor formula the denominator in $g(t)$ can be
presented in the following form:
$$
Q_0(t_0)-Q_0(t)=-Q_0''(t_0)(t-t_0)^2/2-Q_0'''(t_0)(t-t_0)^3/6-Q_0^{(4)}(\zeta)(t-t_0)^4/24,
$$
where $\zeta$ belongs to the interval $(t,\,t_0)$.
If we denote $c_3=Q_0'''(t_0)/6$, we obtain
$$
g(t_1)g(t_2)>\frac{1}{(2-q-c_3(t_1-t_0))(2-q-c_3(t_2-t_0))}
$$
Denominator of this fraction
$$
(2-q)^2+c_3(2-q)(2t_0-t_1-t_2)+c_3^2(t_1-t_0)(t_2-t_0)
$$
has its fist term positive and its second and third terms negative.
Indeed, we have $Q_0(t_1)=Q_0(t_2)$, and, using the Taylor formula
around $t=t_0$ for both parts of this equality, we obtain
$$
(q-2)(t_1-t_0)^2+Q_0'''(\xi_1)(t_1-t_0)^3/6=(q-2)(t_2-t_0)^2+Q_0'''(\xi_2)(t_2-t_0)^3/6,
$$
where $\xi_1\in(t_1,t_0)$ and $\xi_2\in(t_0,t_2)$. Then
$$
(q-2)((t_1-t_0)^2-(t_2-t_0)^2)=Q_0'''(\xi_2)(t_2-t_0)^3/6-Q_0'''(\xi_1)(t_1-t_0)^3/6\leq
0,
$$
since $Q_0'''(t)<0$ for any $t$. Which implies that
$(t_1-t_2)(t_1+t_2-2t_0)$ from which follows that the second term is
negative. Using this fact we obtain
$$
g(t_1)g(t_2)>\frac{1}{(2-q)^2},\quad
I_3>\pi\frac{q}{2\sqrt{2-q}}=\frac{\pi}{\sqrt{p}}\sqrt{\frac{2}{p+2}}\,,
$$
$$
I(c)\geq
\frac{1}{\tau_0}(I_3)^\frac{2}{p}>\frac{1}{\tau_0}\left(\frac{\pi^2}{p}\right)^\frac{1}{p}\left(\frac{2}{p+2}\right)^\frac{1}{p}
$$
which completes the proof.
\end{proof}

\paragraph{Acknowledgments.} The authors are grateful to Dr. Yu. Semenov for the help with the proof of Lemma \ref{l4:2}. The research of ASN is supported by the Department of Health and Human Services intramural program (NIH, National Library of Medicine).

\end{document}